\documentclass[aps,pra,twocolumn,amsmath,amssymb]{revtex4-2}

\usepackage{mathtools}
\usepackage{amsthm}
\usepackage{bbm}
\usepackage[dvipsnames]{xcolor}
\usepackage{algorithm,algpseudocode}
\usepackage[hidelinks]{hyperref}
\hypersetup{colorlinks=true,linkcolor=MidnightBlue,filecolor=MidnightBlue,citecolor=NavyBlue,urlcolor=black}
\usepackage{tikz}
\usetikzlibrary{quantikz,tikzmark,shapes.multipart,fit,positioning,shapes.geometric,arrows}

\newtheorem{theorem}{Theorem}
\newtheorem{proposition}{Proposition}

\newcommand{\R}{\mathbb{R}}
\newcommand{\ra}{\rightarrow}
\newcommand{\rn}{\{0,1\}}

\newcommand\ind[1]{\mathbbm{1}_{#1}}
\DeclareMathOperator{\polylog}{polylog}
\DeclareOldFontCommand{\sc}{\normalfont\scshape}{\mathsc}

\let\bra\relax 
\let\ket\relax
\DeclarePairedDelimiter{\pt}()
\DeclarePairedDelimiter{\bc}[]
\DeclarePairedDelimiter{\set}{\{}{\}}
\DeclarePairedDelimiter{\abs}{\lvert}{\rvert}
\DeclarePairedDelimiter{\ket}{|}{\rangle}
\DeclarePairedDelimiter{\bra}{\langle}{|}

\newcommand{\bo}{O\pt}
\newcommand{\om}{\Omega\pt}
\newcommand{\ta}{\Theta\pt}

\newcommand{\qub}{\ket{\vec{0}}}

\newcommand{\proc}[1]{\textup{\textsf{#1}}}
\newcommand{\ora}{\mathcal{O}}
\newcommand{\ci}{\mathcal{C}}
\newcommand{\di}{\mathcal{D}}

\newcommand{\wn}{\mathcal{Z}}
\newcommand{\wnorm}{\mathcal{W}} 
\newcommand{\bw}{w}
\newcommand{\ketw}{\ket{w}}
\newcommand{\ketwp}{\ket{w^\perp}}
\newcommand{\hi}{h} 
\newcommand{\hiS}{H} 
\newcommand{\rotg}{\mathrm{Rot}}
\newcommand{\qry}{\proc{qry}}
\newcommand{\rot}{\proc{rot}}
\newcommand{\otp}{\proc{out}}
\newcommand{\indi}{\proc{ind}}

\begin{document}

\title{Preparing many copies of a quantum state in the black-box model}
\author{Yassine Hamoudi}
\affiliation{Simons Institute for the Theory of Computing, Berkeley, California 94720, USA}
\email{hamoudi@berkeley.edu}
\date{\today}

\begin{abstract}
  We describe a simple quantum algorithm for preparing $K$ copies of an $N$-dimensional quantum state whose amplitudes are given by a quantum oracle. Our result extends a previous work of Grover, who showed how to prepare one copy in time $\bo{\sqrt{N}}$. In comparison with the naive $\bo{K\sqrt{N}}$ solution obtained by repeating this procedure~$K$ times, our algorithm achieves the optimal running time of $\ta{\sqrt{KN}}$. Our technique uses a refinement of the quantum rejection sampling method employed by Grover. As a direct application, we obtain a similar speed-up for obtaining $K$ independent samples from a distribution whose probability vector is given by a quantum oracle.
\end{abstract}

\maketitle


\section{Introduction}

The preparation of a specific quantum state is an important building block and a critical bottleneck in many quantum algorithms~\cite{MNRS11j,HHL09j,Ber14j,Aar15j}. The objective of the {\sc State Preparation} problem is to find the minimum amount of resources needed to generate a quantum state given some description of it. In general, the complexity of this problem scales linearly with the dimension of the state to be prepared~\cite{MVBS05j,PB11j}. Yet, it is possible to achieve sublinear bounds for particular states or input models. One such example is the \emph{black-box model} where, given oracle access to a non-zero non-negative vector $\bw = (w_1,\dots,w_N)$, the objective is to load the associated normalized probability vector into the amplitudes of the $\log(N)$-qubit state~$\ketw$ defined as
  \begin{equation}
    \ketw := \frac{1}{\sqrt{\wnorm}} \sum_{i=1}^N \sqrt{w_i} \ket{i}
  \end{equation}
where $\wnorm = \sum_{i=1}^N w_i$ is the (unknown) normalization factor. Grover adapted his celebrated quantum search algorithm to this problem in~\cite{Gro00j}, where he showed that $\bo{\sqrt{N}}$ queries to $\bw$ are sufficient to prepare $\ketw$. In practice, one can expect that \emph{several} copies of the same quantum state are needed for further use. For instance,~$\ketw$ may be fed in an algorithm that fails with some probability and that must be repeated several times. The no-cloning theorem prevents the state $\ketw$ from being easily duplicated. In fact, it is easy to show that additional queries to the input are required to prepare several copies of $\ketw$. The problem of adapting the state preparation procedure to the desired number $K$ of copies has received little attention. Usually, it is possible to simply repeat $K$ times the procedure used to prepare one copy, but the complexity grows linearly with $K$. For instance, the algorithm of Grover leads to a query complexity of $\bo{K \sqrt{N}}$ for preparing the $K$-fold state~$\ketw^{\otimes K}$. In this paper, we investigate the question of whether a more efficient approach exists. We describe a two-phase algorithm consisting of a preprocessing step that uses $\bo{\sqrt{KN}}$ queries, after which each copy of~$\ketw$ requires only $\bo{\sqrt{N/K}}$ queries to be prepared. Our result improves upon the previous approach by a factor of $\sqrt{K}$, and it is shown to be optimal.


\subsection{Related work}

Our work is based on the \emph{quantum rejection sampling} method, where a state that is easy to prepare (e.g.\ a uniform superposition) is mapped to a target state by amplitude amplification. This method was pioneered by Grover in~\cite{Gro00j} and subsequently studied in~\cite{ORR13j,SLSB19j,LYC14j,WG17j}. All of these works (except for~\cite{LYC14j}) take place in the quantum oracle model and they often require a number of queries that is polynomial in the dimension $N$ of the state. In the non-oracular setting, the problem of loading an arbitrary vector $(w_1,\dots,w_N)$ into the amplitudes of a quantum state can be done with a circuit of depth~$\bo{N}$ and width $\bo{\log N}$~\cite{MVBS05j,PB11j}. It is possible to use only $\polylog(N)$ resources for specific cases such as efficiently integrable probability distributions~\cite{Zal98j,GR02p,KM01c} (Proposition~\ref{Prop:StateInt}), uniformly bounded amplitudes~\cite{SS06j}, Gaussian states~\cite{KW09p} or probability distributions resulting from a Bayesian network~\cite{LYC14j}. A different line of work~\cite{AT07j,SBBK08j,WA08j,OBD18j,Ape19c,HW20c} studied the preparation of a quantum state that corresponds to the stationary distribution of a Markov chain. These algorithms use Markov chain Monte Carlo methods and quantum walk techniques to obtain a preparation time scaling with the spectral gap. Aharonov and Ta-Shma~\cite{AT07j} also showed that the existence of an efficient procedure to convert any circuit into a coherent state encoding the output distribution of that circuit would imply that $\proc{SZK} \subseteq \proc{BQP}$.

The {\sc State Preparation} problem is also related to the task of preparing samples from a discrete distribution. We refer the reader to~\cite{Dev86b,BFS87b,Knu98b} for a general introduction on the latter topic. In particular, the {\sc Importance Sampling} problem (also called {\sc Weighted Sampling} or {\sc $L_1$ Sampling}) asks to generate~$K$ independent samples from the probability vector $\pt[\big]{\frac{w_1}{\wnorm},\dots,\frac{w_N}{\wnorm}}$ associated with a non-negative weight vector $w$. The \emph{alias method}~\cite{Wal74j,Wal77j,KP79j,Vos91j} solves this problem with a preprocessing cost of~$\bo{N}$ operations, and a generating cost of~$\bo{1}$ operations per sample. Grover~\cite{Gro00c,Gro00j} suggested a quadratically faster algorithm for obtaining one sample, based on preparing the state $\ketw$ and measuring it in the computational basis. Our state preparation algorithm extends the work of Grover to the case of generating $K$ independent samples with a total cost of $\bo{\sqrt{KN}}$ operations. An alternative quantum algorithm for (approximately) generating $K$ such samples was proposed before in~\cite{HRRS19j}, where it was combined with the stochastic gradient descent method to address the submodular function minimization problem.


\subsection{Overview}

The two parts of our state preparation algorithm are described in Theorems~\ref{Thm:StatePreprop} and~\ref{Thm:StatePrepa}. Combined together, these results lead to the following main theorem.

\begin{theorem}
  There is a quantum algorithm with the following properties. Consider two integers $1 \leq K \leq N$, a real $\delta \in (0,1)$ and a non-zero vector~$\bw \in \R^N_{\geq 0}$. Then, with probability at least~$1-\delta$, the algorithm outputs $K$ copies of the state $\ketw$ and it uses $\bo{\sqrt{KN} \log(1/\delta)}$ queries to~$\bw$ in expectation.
\end{theorem}

We now provide a high-level description of how the algorithm works. Our starting point is the result from Grover~\cite{Gro00j} for preparing one copy of the state $\ketw$ in the black-box model. Given an upper bound $\hi$ on the largest value $\max_i w_i$, this algorithm uses two queries and one controlled rotation to implement a unitary $U$ such that,
  \begin{equation}
    U \qub = \frac{1}{\sqrt{N}} \sum_{i=1}^N \ket{i} \pt[\bigg]{\sqrt{\frac{w_i}{\hi}} \ket{0} + \sqrt{1 - \frac{w_i}{\hi}} \ket{1}}.
  \end{equation}
The state $\ketw\ket{0}$ has amplitude $\sqrt{\frac{\wnorm}{N \hi}}$ in $U \qub$, thus it can be extracted by using the amplitude amplification algorithm with $\bo[\Big]{\sqrt{\frac{N \hi}{\wnorm}}}$ applications of $U$ and $U^{-1}$. In particular, if the largest coordinate of $\bw$ is smaller than $\hi = \wnorm / K$ then we can prepare one copy of $\ketw$ in time $\bo{\sqrt{N/K}}$, and $K$ copies in time $\bo{\sqrt{KN}}$. We use this observation to construct a new circuit~$\ci$ (Figure~\ref{Fig:stateCirc}) such that $\ketw \ket{0}$ has amplitude at least~$\sqrt{K/N}$ in~$\ci \qub$, \emph{even} if~$\bw$ contains large coordinates. This circuit uses only two queries to $\bw$, but it requires $\bo{\sqrt{KN}}$ queries to be constructed during a preprocessing phase that is executed only once (Theorem~\ref{Thm:StatePreprop}). The preprocessing phase consists first of computing the set $H$ that contains the positions of the~$K$ largest coordinates in~$\bw$, by using a variant of the quantum maximum finding algorithm (Proposition~\ref{Prop:topK}). The circuit $\ci$ is then defined to proceed in two stages. First, it prepares a state whose amplitudes depend only on the values in $\set{w_i : i \in H}$ (Equation~(\ref{Eq:stateFirst})). Next, it modifies this state by querying the set $\set{w_i : i \notin H}$ in a way that is similar to that of $U$. The crucial observation is that the values in $\set{w_i : i \notin H}$ must be smaller than $\wnorm / K$ by definition of~$H$, thus they can be amplified at a smaller cost. Finally, each copy of~$\ketw$ can be obtained by one application of the amplitude amplification algorithm on $\ci$ (Theorem~\ref{Thm:StatePrepa}).

We show in the next proposition that our algorithm is optimal by a simple reduction from the {\sc $K$-Search} problem.

\begin{proposition}
  \label{Prop:lowerSampl}
  Any bounded-error quantum algorithm that can output $K$ copies of the quantum state $\ketw$ given oracle access to any non-zero vector $\bw \in \R^N_{\geq 0}$ must perform at least $\om[\big]{\sqrt{KN}}$ quantum queries to $\bw$.
\end{proposition}

\begin{proof}
  We consider a variant of the {\sc $K$-Search} problem where the objective is to find $K$ preimages of $1$ in an oracle $\ora : [N] \ra \rn$ containing at least $2K$ such preimages. The bounded-error quantum query complexity of this problem is $\ta{\sqrt{KN}}$ (the proof can easily be derived from~\cite[Appendix A]{HM21c} for instance). On the other hand, by a coupon collector argument~\cite{MiscCoupon17}, if we prepare and measure in the computational basis $\ta{K}$ copies of $\ketw$ where $\bw = \pt{\ora(1),\dots,\ora(N)} \in \rn^N$, then we obtain the positions of at least~$K$ different preimages of~$1$ with constant success probability. It implies that generating~$\om{K}$ such copies requires using at least $\om{\sqrt{KN}}$ quantum queries to~$\bw$.
\end{proof}


\section{Preliminaries}

\subsection{Computational model}
\label{Sec:modelState}

We use the quantum circuit model over a universal gate set made of the CNOT gate and of all one-qubit gates. We suppose that the real numbers manipulated by our algorithms (such as the coordinates of $\bw$) can be encoded over~$c$ bits, for a fixed value of~$c$. In particular, these numbers can be stored in quantum registers of size~$c$. We also add the three gates described in Figure~\ref{Fig:BasicGates}. The indicator gate~$\ind{\hiS}$ is specified by a subset $\hiS \subseteq [N]$. It operates on a Boolean value $b$ and an index $i \in [N]$. The Boolean value $\bc{i \notin \hiS}$ is equal to~$1$ if and only if $i \notin \hiS$. The query gate~$\ora_{\bw}$ is specified by the input vector~$\bw$ to the problem. It operates on an index $i \in [N]$ and a real~$v$ (encoded over~$c$ bits). Finally, the controlled rotation gate $\rotg_{\hi}$ is specified by a real $\hi > 0$ and it operates on a Boolean value $b$ and a real $v$. We refer the reader to~\cite{HRS18p,SLSB19j} and references therein for efficient implementations of the arithmetic gates and controlled rotation gates with a given precision. We will also use~$\qub$ in our notations to represent a multi-qubit basis state~$\ket{0}^{\otimes \ell}$ for some $\ell > 1$.

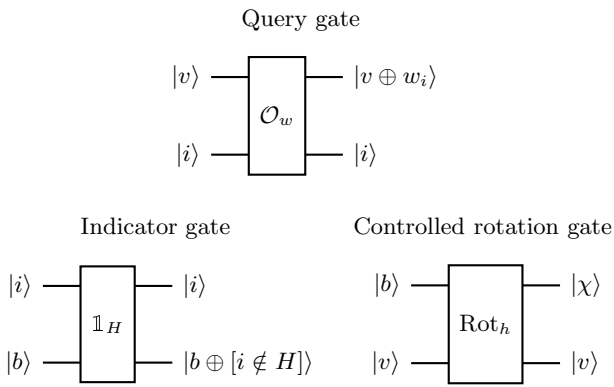
\begin{figure}[htb] 
  \hspace{-5mm}
  \sbox0{\begin{quantikz}
    \lstick{$\ket{i}$} & \gate[wires=2]{\ind{\hiS}} & \qw \rstick{$\ket{i}$} \\
    \lstick{$\ket{b}$} &  & \qw \rstick{$\ket{b \oplus \bc{i \notin \hiS}}$}
  \end{quantikz}}%
  \sbox1{\begin{quantikz}
    \lstick{$\ket{v}$} & \gate[wires=2]{\ora_{\bw}} & \qw \rstick{$\ket{v \oplus w_i}$} \\
    \lstick{$\ket{i}$} &  & \qw \rstick{$\ket{i}$}
  \end{quantikz}}%
  \sbox2{\begin{quantikz}
    \lstick{$\ket{b}$} & \gate[wires=2]{\rotg_{\hi}} & \qw \rstick{$\ket{\chi}$} \\
    \lstick{$\ket{v}$} &  & \qw \rstick{$\ket{v}$}
  \end{quantikz}}%
  \setlength{\tabcolsep}{6pt}
  \begin{tabular}{cc}
    \multicolumn{2}{c}{Query gate} \\[2mm]
    \multicolumn{2}{c}{\usebox1} \\[12mm]
    Indicator gate & Controlled rotation gate \\[2mm]
    \usebox0 & \usebox2
  \end{tabular}
  \caption{Three gates used in our circuits. The state $\ket{\chi}$ is defined as $\sqrt{\frac{v}{\hi}} \ket{b} + \sqrt{1-\frac{v}{\hi}}\ket{1-b}$ if $0 < v \leq \hi$, and $\ket{b}$ otherwise.}
  \label{Fig:BasicGates}
\end{figure}

The \emph{query complexity} of an algorithm is the number of times it uses the oracle gate~$\ora_w$ to access the input~$w$. If not specified, the total number of elementary gates used by our algorithms will be larger than their respective query complexity by at most a $\polylog(N)$ factor.


\subsection{Quantum subroutines}
\label{Sec:prelimState}

The next algorithm generalizes the quantum minimum finding of D{\"{u}}rr and H{\o}yer~\cite{DH96p} to finding $K$ largest entries in a vector $\bw$. The algorithm succeeds if it outputs a set $\hiS \subset [N]$ of $K$ coordinates that dominate all other entries. There is not necessarily a unique choice for $\hiS$ since different coordinates of $\bw$ may be equal.

\begin{proposition}[{\sc Top-$K$ maximum finding} -- Theorem~4.2 in \cite{DHHM06j}]
  \label{Prop:topK}
  There exists a quantum algorithm with the following properties. Consider two integers $1 \leq K \leq N$, a real $\delta \in (0,1)$ and a vector $\bw \in \R^N_{\geq 0}$. Then, the algorithm outputs the positions of $K$ largest entries in~$\bw$ with success probability at least $1-\delta$, and it performs $\bo{\sqrt{KN} \log(1/\delta)}$ queries to $\bw$.
\end{proposition}

We also need the well-known amplitude amplification algorithm. 

\begin{proposition}[{\sc Amplitude Amplification} -- Theorem 3 in \cite{BHMT02j}]
  \label{Prop:AA}
  Let $\ci$ be a quantum circuit that prepares the state $\ci \qub = \sqrt{p} \ket{\varphi}\ket{0} + \sqrt{1-p} \ket{\varphi^{\perp}}\ket{1}$ for some $p \in [0,1]$ and two unit states $\ket{\varphi}$, $\ket{\varphi^{\perp}}$. Then, the \emph{amplitude amplification} algorithm outputs the state $\ket{\varphi}$ by using~$\bo{1/\sqrt{p}}$ applications of $\ci$ and $\ci^\dagger$ in expectation.
\end{proposition}

Finally we will use the next quantum state preparation algorithm that requires having an efficient procedure to compute the partial sum $\sum_{\ell=i}^j w_\ell$ for any $1 \leq i \leq j \leq N$.

\begin{proposition}[\sc State preparation by integration -- \cite{Zal98j,KM01c,GR02p}]
  \label{Prop:StateInt}
  There is a quantum algorithm with the following properties. Consider an integer $N$ and a non-zero vector $\bw \in \R^N_{\geq 0}$ such that there is a (classical) reversible circuit with~$T$ gates that computes $\sum_{\ell=i}^j w_\ell$ given $i \leq j$. Then, the algorithm outputs $\ketw$ and it uses $\bo{T \log N}$ elementary gates.
\end{proposition}


\section{Main algorithm}
\label{Sec:algo}

We describe in details our state preparation algorithm for preparing $K$ copies of $\ketw$ given oracle access to $\bw = (w_1,\dots,w_N)$. The first step of the algorithm is a preprocessing phase (Algorithm~\ref{Alg:StatePreprop}) that constructs a particular circuit $\ci$ described in Figure~\ref{Fig:stateCirc}.

\begin{algorithm}[H]
  \caption{Preprocessing phase.}
    \label{Alg:StatePreprop}
  \begin{algorithmic}[1]
    \State Compute a set $\hiS \subseteq [N]$ of the positions of $K$ largest entries in $\bw$ by using the top-$K$ maximum finding algorithm (Proposition~\ref{Prop:topK}) with failure probability~$\delta$.
    \State Compute $\hi = \min_{i \in \hiS} w_i$ and
             \begin{equation}
               \wn = (N-K)\hi + \sum_{i \in \hiS} w_i.
             \end{equation}
    \State Use the state preparation algorithm of Proposition~\ref{Prop:StateInt} to construct a circuit~$\di$ such that, on input $\qub_{\otp}$, it prepares the state
         \begin{equation}
           \label{Eq:stateFirst}
           \di \qub_{\otp} = \sum_{i \in \hiS} \sqrt{\frac{w_i}{\wn}} \ket{i}_{\otp} + \sum_{i \notin \hiS} \sqrt{\frac{\hi}{\wn}} \ket{i}_{\otp}.
         \end{equation}
    \State Output the circuit $\ci$ represented in Figure~\ref{Fig:stateCirc}.
  \end{algorithmic}
\end{algorithm}

\begin{figure}[htb] 
  \centering
  \begin{quantikz}[transparent]
    \lstick{$\ket{0}_{\rot}$} &[-2mm] \qw &[-2mm] \qw &[-2mm] \qw &[-2mm] \gate[2]{\rotg_{\hi}} &[-2mm] \qw &[-2mm] \qw &[-2mm] \qw \\[-0.3cm]
    \lstick{$\qub_{\qry}$} & \qw & \qw & \gate[wires=2]{\ora_{\bw}} & \qw & \gate[wires=2]{\ora_{\bw}} & \qw & \qw \\[-0.3cm]
    \lstick{$\qub_{\otp}$} & \gate[1][0.7cm][0.7cm]{\di} & \gate[wires=2]{\ind{\hiS}} & \qw & \qw & & \gate[wires=2]{\ind{\hiS}} & \qw \\[-0.3cm]
    \lstick{$\ket{0}_{\indi}$} & \qw &  & \qw & \ctrl{-2} & \qw & & \qw
  \end{quantikz}
  \caption{Circuit $\ci$ output at the end of the preprocessing phase (Algorithm~\ref{Alg:StatePreprop}).}
  \label{Fig:stateCirc}
\end{figure}
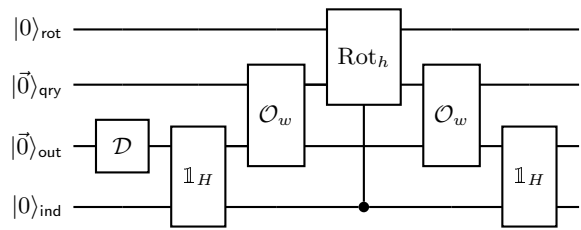

\begin{theorem}[\sc Preprocessing phase]
  \label{Thm:StatePreprop}
  Consider two integers $1 \leq K \leq N$, a real $\delta \in (0,1)$ and a non-zero vector $\bw \in \R^N_{\geq 0}$. Then, Algorithm~\ref{Alg:StatePreprop} outputs with probability at least $1-\delta$ the description of a quantum circuit~$\ci$ such that, on input $\qub$, it prepares the state
    \begin{equation}
      \ket{\psi} = \sqrt{p_{\bw}} \ketw \ket{0} + \sqrt{1 - p_{\bw}} \ketwp \ket{1}
    \end{equation}
  where $p_{\bw} \geq K/N$ and $\ketwp$ is some unit state. The algorithm performs $\bo{\sqrt{KN} \log(1/\delta)}$ queries to~$\bw$. The circuit $\ci$ performs two queries to $w$ and it uses $\bo{K \log N}$ elementary gates.
\end{theorem}

\begin{proof}
  We assume that the top-$K$ maximum finding algorithm returns a correct set $\hiS$ at step~1 of Algorithm~\ref{Alg:StatePreprop}, which is the case with probability at least $1-\delta$. The circuit $\ci$ represented in Figure~\ref{Fig:stateCirc} operates on four registers: $\rot$ and $\indi$ that contain a Boolean value,~$\qry$ that contains a real $v \geq 0$, and~$\otp$ that contains an integer $i \in [N]$. The indicator gate $\ind{\hiS}$ flips the content of $\indi$ when the $\otp$ register contains $i \notin \hiS$, which allows the rotation gate $\rotg_{\hi}$ to be activated only when $i \notin \hiS$. The gates~$\ora_{\bw}$ and~$\ind{\hiS}$ are applied a second time at the end of $\ci$ to uncompute the registers~$\qry$ and~$\indi$. A simple calculation shows that the final state is $\ci \qub = \qub_{\qry}\ket{0}_{\indi}\ket{\psi}_{\otp,\rot}$ where
  \begin{equation}
    \begin{aligned}
      \ket{\psi}_{\otp,\rot}
        & = \sum_{i \in \hiS} \sqrt{\frac{w_i}{\wn}} \ket{i}_{\otp} \ket{0}_{\rot} \\
        & + \sum_{i \notin \hiS} \sqrt{\frac{\hi}{\wn}} \ket{i}_{\otp} \pt*{\sqrt{\frac{w_i}{\hi}} \ket{0}_{\rot} + \sqrt{1 - \frac{w_i}{\hi}} \ket{1}_{\rot}} \\
        & = \sqrt{\frac{\wnorm}{\wn}} \ketw_{\otp} \ket{0}_{\rot} + \sqrt{1 - \frac{\wnorm}{\wn}} \ketwp_{\otp} \ket{1}_{\rot}
    \end{aligned}
  \end{equation}
  for some irrelevant unit state $\ketwp_{\otp}$. In order to lower bound the coefficient $p_{\bw} := \frac{\wnorm}{\wn}$, we first observe that the smallest value $\hi = \min_{i \in \hiS} w_i$ over $\hiS$ must satisfy $\hi \leq \frac{\wnorm}{K}$ since otherwise $\sum_{i \in \hiS} w_i$ would exceed $\wnorm$. Thus, $p_{\bw}^{-1} = \frac{\wn}{\wnorm} = \frac{(N-K)\hi + \sum_{i \in \hiS} w_i}{\wnorm} \leq \frac{N-K}{K} + 1 = \frac{N}{K}$.

  The algorithm uses $\bo{\sqrt{KN} \log(1/\delta)}$ queries at step~1 by Proposition~\ref{Prop:topK}. The set $\set{\bw_i : i \in H}$ can be computed with $K$ queries, after which steps 2--4 do not need to perform any new query. For any $i \leq j$, the partial amplitude sum $\sum_{\ell = i}^j \bra{\ell} \di \qub^2$ is equal to $\frac{1}{\wn} \pt[\big]{\sum_{\ell \in \hiS \cap \set{i,\dots,j}} w_i + \hi \cdot \abs{\set{i,\dots,j} \setminus \hiS}}$, which can be computed by a classical circuit with $\bo{K}$ gates since~$\hiS$ is of size $K$. Thus, by Proposition~\ref{Prop:StateInt}, the circuit~$\di$ constructed at step 3 requires $\bo{K \log N}$ elementary gates. Finally, the number of gates needed to implement the circuit $\ci$ at step 4 is dominated by the number of gates needed in $\di$ since the other gates are included in the computational model (see Section~\ref{Sec:modelState}).
\end{proof}

We use the circuit $\ci$ constructed during the above preprocessing phase, together with the amplitude amplification algorithm, to obtain the next state preparation phase that generates one copy of~$\ketw$ in time $\bo{\sqrt{N/K}}$.

\begin{theorem}[\sc State preparation phase]
  \label{Thm:StatePrepa}
  Consider two integers $1 \leq K \leq N$, a real $\delta \in (0,1)$ and a non-zero vector $\bw \in \R^N_{\geq 0}$. Let $\ci$ denote a quantum circuit obtained with Algorithm~\ref{Alg:StatePreprop} on input $K$, $\delta$, $\bw$ that correctly prepares the state $\ket{\psi}$ described in Theorem~\ref{Thm:StatePreprop}. Then, given the description of $\ci$, one can prepare the state $\ketw$ by using $\bo{\sqrt{N/K}}$ queries to $\bw$ in expectation.
\end{theorem}

\begin{proof}
  This is a direct application of the amplitude amplification algorithm (Proposition~\ref{Prop:AA}) on~$\ci$, where the complexity is derived from the fact that $\ketw\ket{0}$ has amplitude at least $\sqrt{K/N}$ in $\ket{\psi}$ by Theorem~\ref{Thm:StatePreprop}.
\end{proof}


\section{Discussion}

We did not address the precision errors in our analysis. In particular, it can be relevant to replace the controlled rotation gate (which requires to calculate the arcsine function) by the comparison-based circuit defined in~\cite{SLSB19j} that avoids arithmetic. We also restricted ourselves to preparing states with non-negative real amplitudes. Arbitrary phase factors can be introduced by using the techniques discussed in~\cite{KM01c,SLSB19j}.

\bibliography{Bibliography}

\end{document}